\documentclass{article}

\usepackage[utf8]{inputenc} 
\usepackage[T1]{fontenc}    
\usepackage{hyperref}       
\usepackage{url}            
\usepackage{booktabs}       
\usepackage{amsfonts}       
\usepackage{nicefrac}       
\usepackage{microtype}      
\usepackage{xcolor}         
\usepackage{geometry}
\usepackage{amsmath, amsthm, amstext}
\usepackage{cleveref}
\usepackage{dsfont}
\usepackage{subcaption}
\usepackage{graphicx}
\title{V3rified: Revelation vs Non-Revelation Mechanisms for Decentralized Verifiable Computation}

\author{Tiantian Gong \thanks{Yale University. Email: tiantian.gong@yale.edu}\and Aniket Kate \thanks{Purdue University \& Supra Research. Email: aniket@purdue.edu} \and Alexandros Psomas\thanks{Purdue University. Email: apsomas@purdue.edu} \and Athina Terzoglou\thanks{Purdue University. Email: aterzogl@purdue.edu}}

\date{}

\DeclareMathOperator*{\argmin}{arg\,min}
\DeclareMathOperator*{\argmax}{arg\,max}
\newtheorem{theorem}{Theorem}[section]
\newtheorem{definition}{Definition}[section]

\newtheorem{proposition}{Proposition}[section]

\newtheorem{lemma}{Lemma}[section]
\newtheorem{claim}{Claim}[section]

\newtheorem{example}{Example}[section]

\allowdisplaybreaks

\newcommand{\Rule}{\mathcal{R}}
\newcommand{\ks}{{k^*}}
\newcommand{\Mech}{\mathcal{M}}
\newcommand{\bid}{\mathbf{b}}

\begin{document}

\maketitle

\begin{abstract}
In the era of Web3, decentralized technologies have emerged as the cornerstone of a new digital paradigm. Backed by a decentralized blockchain architecture, the Web3 space aims to democratize all aspects of the web. From data-sharing to learning models, outsourcing computation is an established, prevalent practice.
 Verifiable computation makes this practice trustworthy as clients/users can now efficiently validate the integrity of a computation. As verifiable computation gets considered for applications in the Web3 space, decentralization is an important requirement that guarantees system reliability and ensures that no single entity can suppress clients.  At the same time, decentralization needs to be balanced with efficiency: clients want their computations done as quickly as possible.

Motivated by these issues, we study the trade-off between decentralization and efficiency when outsourcing computational tasks to strategic, rational solution providers. Specifically, we examine this trade-off when the client employs (1) revelation mechanisms, i.e. auctions, where solution providers bid their desired reward for completing the task by a specific deadline and then the client selects which of them will do the task and how much they will be rewarded, and (2) simple, non-revelation mechanisms, where the client commits to the set of rules she will use to map solutions at specific times to rewards and then solution providers decide whether they want to do the task or not. We completely characterize the power and limitations of revelation and non-revelation mechanisms in our model.
\end{abstract}

\section{Introduction}\label{sec:introduction}
Over the last twenty years, computation outsourcing has grown enormously in the Web2 space. For example, the cloud computing market size is projected to be valued at \$980 billion in 2025~\cite{statsweb2}. 
Meanwhile, computation outsourcing in Web3 has also been steadily increasing since 2023, with the market of zero-knowledge proof computation outsourcing hitting \$609 million in June 2024~\cite{statsweb3} and verifiable on-(block)chain randomness service seeing more than a million requests every month~\cite{delphidigital}.

Unlike centralized cloud computing in Web2, Web3 prioritizes \emph{decentralization}, ensuring that the system does not rely on a single party to complete a task. This reduces single points of failure of the center, such as infrequent but catastrophic crashes or compromises, selective censorship of client requests~\cite{BlockchainCensorship}, determining task completion time at its will, or charging high prices. Through decentralization, no single provider can monopolize any solution platform. Aside from maintaining decentralization, 
considering that these computation tasks often come with time constraints (otherwise clients could perform the computations locally), \emph{efficiency}, i.e., clients acquiring solutions as fast as possible, is a crucial property for computation outsourcing platforms. 

Without centralized trust, clients naturally desire to verify computation results. 
Verifiable computation~\cite{gennaro2010non,parno2016pinocchio} allows a computationally limited client to delegate a computation task to more powerful and resourceful servers, while ensuring that the output can be verified efficiently. 
Designing mechanisms that incentivize decentralization in Web3 is still an active field. 
Current platforms for verifiable computation treat decentralization as given. For example, 
Gevulot~\cite{gevulot} and Taiko~\cite{taiko} randomly select a worker for a task submitted by a client and pay the worker a posted price.
Supra dVRF~\cite{SupraVRF} makes the {\em any-trust} assumption: on every set of workers, at least one of them behaves honestly. Assuming that there exist honest workers who always follow the protocol (even if they suffer monetary losses) and never crash, such mechanisms do avoid centralization. However, given the real-world monetary rewards, it is, of course, reasonable to expect that these supposedly honest workers are \emph{rational} and take actions that maximize their utility; decentralization cannot be taken for granted.

Mechanism design, which studies how to craft a set of rules so that self-interested agents who act upon these rules arrive at a desired (for the designer) outcome, has affected a wide range of applications in Web3 space, including pricing oracle services like randomness oracles and off-chain data oracles~\cite{oracle_risks,oracle_survey} and zero-knowledge proof markets~\cite{statsweb3}. The theory of mechanism design has almost exclusively focused on so-called \emph{revelation} mechanisms. Such mechanisms are designed so that agents' optimal action is very simple: ``report your honest preferences.'' The designer has the complex task of finding outcomes in a way that truthfulness, as well as other guarantees (e.g. good revenue in the case of auctions), are satisfied. The focus on revelation mechanisms has been motivated by the \emph{revelation principle,} a fundamental result in mechanism design which suggests that if some outcome is the equilibrium of a non-truthful mechanism, then one can construct a truthful mechanism to implement this outcome. 
However, the restriction to revelation mechanisms often comes at a loss, e.g. in non-Bayesian or prior-free settings~\cite{feng2018end,feng2021revelation}. Furthermore, non-truthful or \emph{non-revelation mechanisms}, e.g., the first-price auction, because of their simplicity, are a lot more common in practice, including the aforementioned case of computation outsourcing in Web3.

In this paper, we are interested in comparing revelation and non-revelation mechanisms. Simply put, our goal is \emph{to characterize the trade-offs between decentralization and efficiency for both revelation and non-revelation mechanisms, in the context of outsourcing computation in Web3}.

\subsection{Our contribution}

Motivated by real-world computation outsourcing markets (see Related work), we introduce a simple model for studying decentralization and efficiency in the presence of strategic behavior. There is a client, with a total reward of $R$, that needs to outsource a computational task, which must be completed before a deadline $T$, and there are $n$ strategic agents/solution providers. Agent $i$ can pay a cost $c_i$ to complete the computational task by time $t_i$; both $c_i$ and $t_i$ are known to the agent and not the client. We assume that agents with higher costs can complete tasks faster.

The client can employ a \emph{revelation} or \emph{non-revelation} mechanism to interact with and incentivize the agents to complete her task. A revelation mechanism first asks agents to report a cost and completion time (essentially, a bid), and then selects an outcome: who does the task (possibly more than one agent) and how much reward they get. A non-revelation mechanism is a set of rules for mapping solutions at specific times to rewards, e.g. ``fastest solution gets the whole reward,'' or ``three fastest solutions split the reward equally.'' For revelation mechanisms, as is standard in mechanism design, we ask for truthfulness and individual rationality (honest agents have non-negative utility when participating); non-revelation mechanisms will be evaluated at their worst pure Nash equilibria.\footnote{Noting that pure Nash equilibria are not guaranteed to exist for arbitrary games; in our positive results we prove that the non-revelation mechanisms we propose do have pure Nash equilibria.} Agents can take actions to maximize their utility (reward minus cost), but cannot misbehave in arbitrary ways. Namely, agent $i$ cannot complete a task faster than time $t_i$.

Intuitively, an outcome is decentralized if many agents submit solutions. Of course, costs can be adversarially selected so that the reward of $R$ does not suffice to cover the cost of even two agents. In the absence of incentives, the number of agents that submit solutions can be maximized by picking the $\ks$ agents with the smallest cost, where $\ks$ is the largest number such that the total cost is at most $R$. We call $\ks$ the \emph{decentralization factor}, and strive to produce $\alpha$-decentralized mechanism outcomes, those where the number of agents submitting solutions is at least $\alpha \ks$. $1$-decentralization implies that a maximum number of $\ks$ agents submit solution, while $\frac{1}{\ks}$-decentralization implies that (in the worst-case) only one agents submits a solution. 

On the other hand, efficient outcomes are those where the first solution is as early as possible. And, again, one can pick costs and times adversarially so that adequately rewarding any two agents to submit solutions (i.e., $\frac{2}{\ks}$-decentralization) is arbitrarily slower than the time the fastest, most expensive agent can produce. Therefore, for an accurate comparison, a better benchmark for efficiency is the best outcome for a \emph{fixed} target decentralization. We say that an outcome is $\alpha$-efficient if it is as fast as the fastest $\alpha$-decentralized outcome. So, intuitively, a $\frac{1}{\ks}$-efficient outcome is one that has a solution at time $\min_{i \in [n]} t_i$, while a $1$-efficient outcome is (in general) slower.

In the absence of strategic behavior, $\alpha$-decentralization is compatible with $\alpha$-efficiency. Since computation is efficiently verifiable, solution providers cannot behave arbitrarily; namely, they cannot submit incorrect solutions to the client's task, as this behavior can be easily detected as dishonest. However, solution providers can report a higher cost, or submit a solution at a later time, if such an action increases their utility.  In this paper, we explore the power and limitations of revelation and non-revelation mechanisms.

In~\Cref{sec:non-revelation} we study non-revelation mechanisms. We prove a tight bound of $1/2$ for the optimal decentralization (Theorems~\ref{thm: half lower bound non-revelation} and~\ref{thm:equalReward}), achieved by the simple non-revelation mechanism that equally splits the reward among all agents that submit a solution. This bound remains tight even for $\ks = 2$. We prove that, as $\ks$ grows, better results are possible. Specifically, it is possible to achieve a decentralization of $1 -\frac{1}{e} - \frac{8}{\ks}$ (which approaches $1 - 1/e \approx 0.63$), with a simple non-revelation mechanism that splits the reward among all agents that submit a solution according to a harmonic sequence, i.e., for some $x$, the fastest solution is rewarded $1/x$, the second fastest $\frac{1}{x+1}$, and so on. Surprisingly, this performance is also asymptotically tight! In~\Cref{thm: strong lower bound non revelation} we prove that no non-revelation mechanism can have decentralization better than $1 - \left( e^{1+\frac{e^2}{2 \ks}} \right)^{-1}$, again, approaching $1- 1/e$. Regarding efficiency, we prove that non-trivial decentralization is incompatible with any non-trivial efficiency guarantee for non-revelation mechanisms (\Cref{thm:efficiency plus decentralization lower bound non revelation}). However, this impossibility can be circumvented by imposing additional structure on the agents' types, such that an agent's true computation time implies her true cost of computation.
In this case, decentralization and efficiency can be combined optimally, as if agents were non-strategic (\Cref{thm: restricted types non revelation}).

In~\Cref{sec: revelation} we study revelation mechanisms. We first show that, with respect to decentralization, revelation mechanisms face the same barriers as non-revelation mechanisms: in an analogue to~\Cref{thm: half lower bound non-revelation}, in~\Cref{thm: revelation decentralization lower bound} we prove that there is no dominant strategy incentive compatibility and individually rational, $\alpha$-decentralized revelation mechanism, for any constant $\alpha > 1/2$.
However, efficiency and decentralization are compatible for revelation mechanisms. We give a mechanism parameterized by $k$ that admits a truthful ex‑post Nash equilibrium and achieves $\frac{\min\{ k, \ks\} -1}{\ks}$-decentralization and $\frac{\min\{ k, \ks\}}{\ks}$-efficiency (\Cref{thm: inverse GSP}). Our mechanism, Inverse Generalized Second Price (I-GSP), resembles the generalized second price auction (GSP), a well-studied \emph{non-truthful mechanism} in auction theory~\cite{edelman2007internet}. In GSP, the $i$-th highest bidder is charged the $(i+1)$-th highest bid. In I-GSP, the agent with the $i$-th smallest cost is rewarded the $(i+1)$-th smallest cost, with a twist to include a fast (and expensive) bidder, in order to guarantee good efficiency.  

In~\Cref{sec: experiments}, we experimentally evaluate our algorithms through two sets of experiments. In the first set, we compare the decentralization factor of the equal reward mechanism $R^{\text{eq}}$ (\Cref{subsec:warm up non revelation}) and the harmonic mechanism $R^{\text{harm}}$ (\Cref{subsec: harmonic}). Our results show that both mechanisms consistently outperform their theoretical worst-case guarantees on these random inputs. In the second set, we compare the fastest solution under the non-revelation harmonic mechanism $R^{\text{harm}}$ with the revelation mechanism inverse generalized second price $\Mech^{\text{I-GSP}}$ (\Cref{sec: revelation}), using the same number of participants. The results confirm that $\Mech^{\text{I-GSP}}$, that prioritizes efficiency, consistently outperforms $R^{\text{harm}}$.

\subsection{Related work}

\textbf{Fee mechanism for computation outsourcing markets in Web3.}
In the proof market of =nil;~\cite{nil} (currently in beta version), a client can submit a transaction specifying the time and price for the request, and whether she desires to prioritize time or price. Solution providers then submit bids, and they are ranked according to their proof generation time and price (in the order specified by the client). The top provider is the selected solver. 
Taiko~\cite{taiko} assigns weights to providers, which are positively correlated with providers' stakes and negatively correlated with their asking prices. For a request, it selects a provider at random according to their weights. 
Wang et al.~\cite{wang2024mechanism} designed a mechanism for prover markets for achieving efficiency (in terms of maximizing social welfare), incentive compatibility, and collusion resistance. The clients are selected via a ``pay-as-bid greedy auction,'' i.e., picking the client transactions that maximize the payment amounts while conforming to a finite batch capacity constraint. The providers are then selected as solvers uniformly at random.  All the above approaches treat decentralization as given. 
Succinct network~\cite{succinct} utilizes an all pay auction~\cite{baye1996all} to allocate proof tasks and rewards: Given $m$ providers, each provider $i$ that bids an amount $b_i$ pays their bid and wins the auction -- thus the prize from the client -- with probability $\frac{b_i^{\alpha}}{\sum_{j=1}^m b_j^{\alpha}}$ where $\alpha$ is a system parameter. 

\textbf{Centralization issues in Web3.}
Heimbach et al.~\cite{heimbach2023ethereum} point out the concentration of Ethereum's relay and builder industry by measuring the Herfindahl-Hirschman Index (HHI) in the two industries. More straightforward statistics also exhibit a centralization trend: The top three MEV-boost~\cite{mevboost} builders in Ethereum can produce more than 90\% of the blocks\cite{mevbuildermarket}; The top five mining pools in Bitcoin can produce approximately 80\% of the blocks~\cite{bitcoinhashrate}.
Previous works~\cite{Gervais13,gencer18,karakostas2022sok} have also more systematically highlighted the concentration of hardware manufacturers, codebases, mining powers, R\&D funding, decision-making, and incident resolution processes in blockchains. 
Wahrstaetter et al.\ \cite{BlockchainCensorship} study censorship by a single entity in prominent permissionless blockchains, and observe that censorship by centralized entities affects not only neutrality but also the security of blockchains.

\textbf{Revelation and non-revelation mechanism design.}  While revelation mechanisms are the dominating paradigm in mechanism design, recent results show that non-revelation mechanisms have surprising theoretical advantages in canonical domains, e.g., prior-independent mechanism design~\cite{feng2018end,feng2021revelation}. Non-revelation mechanisms are prevalent in practice, e.g., the first price auction. A thematically related (to our paper) line of work that studies the design of \emph{contests}, where participants incur a cost and produce an output of a certain quality also, almost exclusively, focuses on non-revelation mechanisms and their equilibria, e.g.,~\cite{Moldovanu2001}; see~\cite{corchon2007theory} for a survey. A major difference with our paper is that contest design ignores decentralization and primarily focuses on maximizing the expected maximum quality or sum of utilities.

\section{Preliminaries}\label{sec: preliminaries}

There is a client interested in outsourcing a computational task that needs to be finished by time $T$, and is willing to provide a reward of $R$; without loss of generality, we normalize this reward to $R = 1$. There are $n$ solution providers that can compute solutions to the client's task; throughout the paper, we refer to them as agents. Agent $i$ has a cost $c_i$ that she can incur to compute a solution at time $t_i \leq T$. We refer to the tuple $(c_i,t_i)$ as agent $i$'s \textit{type}. Agents' types are private. We assume that agents with higher costs can complete tasks faster, i.e., if $c_i>c_j$ then $t_i<t_j$. To ensure, feasibility, we further assume there is at least one agent with $c_i<1$, so that the reward can indeed cover her cost. Agents have quasi-linear utilities; that is, if agent $i$ gets reward $r_i$ and pays a cost $c_i$, her utility is $r_i - c_i$. Agents are rational and will take actions that maximize their utility. For example, in both revelation and non-revelation mechanisms, agents always have the option of not participating, which we will think of as taking the action $\bot$, if their (expected) reward is strictly less than their cost.

\textbf{Revelation and non-revelation mechanisms.} The client can outsource her task by deploying a \emph{revelation}, or a \emph{non-revelation} mechanism.

For non-revelation mechanisms, we consider mechanisms induced by a function that maps a vector of solutions at specific times to reward. We denote with $\bot$ the action of not submitting a solution. A rule $\Rule(a_1,\ldots,a_n) $ outputs a reward vector $(r_1,\ldots,r_n)$, where $a_i$ is the time agent $i$ submitted her solution, or $\bot$ if $i$ did not submit a solution within time $T$. We write $\Rule_i(a_1,\ldots,a_n)$ for the reward of agent $i$ under rules $\Rule$, when the actions taken are $a_1, \ldots, a_n$. Agent $i$ cannot compute solutions in time less than $t_i$. Therefore, we have that for, all $i \in [n]$, $a_i\ge t_i$ (or $a_i = \bot$). Depending on $\Rule$, agent $i$ might maximize her utility by submitting a solution at a time later than $t_i$, or by not submitting a solution at all. We will evaluate non-revelation mechanisms via their (worse) \emph{pure} Nash equilibria. For agents with types $(c_1,t_1), \ldots, (c_n,t_n)$, the strategy profile $(a_1,\ldots,a_n)$ is a pure Nash equilibrium under $\Rule$, if no agent wants to deviate to a different (feasible) action $a'_i$, i.e. for all every agent $i \in [n]$, $\Rule_i(a_i; a_{-i}) - c_i \cdot 1\{ a_i \neq \bot \} \geq \Rule_i(a'_i; a_{-i}) - c_i \cdot 1\{ a'_i \neq \bot \}$, where $1\{ . \}$ is the indicator function and $a_{-i}$ is the set of actions taken by all agents except $i$. We will assume that, if an agent is indifferent between submitting a solution and not submitting a solution, then they prefer submitting a solution. Since we consider (worst-case) pure Nash equilibria, which don't always exist, in our positive results for non-revelation mechanisms we also show pure Nash equilibria exist.

\begin{example}[Reward the fastest solution]\label{example: reward fastest non revelation}
Let $\Rule^{\text{Fast}}$ be the non-revelation mechanism that gives the whole reward to the agent that submitted a solution first, breaking ties lexicographically. Then, $\Rule_i(a_1, \dots, a_n)$ is $1$ for $i = \argmin_{j \in [n]: a_j \neq \bot} \{ a_j \}$, and $0$ otherwise. A pure Nash equilibrium of $\Rule^{\text{Fast}}$ is the (lexicographically in case of a tie) fastest agent, $i^* \in \argmin_{i \in [n]} t_i$, to take action $a_{i^*} = t_i$, and $a_i = \bot$ for all $i \neq i^*$.
\end{example}

A revelation mechanism $\Mech = (x,r)$ consists of two functions: an allocation function $x$ and a reward function $r$, which map \emph{reported} types, or bids, to allocations and rewards. Specifically, both functions take as input agents' bids: the bid $b_i = (\hat{c}_i,\hat{t}_i)$ of agent $i$ consists of a cost $\hat{c}_i$ (possibly different to $c_i$) and a time $\hat{t}_i$ (possibly different to $t_i$). Let $\bid = (b_1,\ldots, b_n)$ denote a vector of bids. $x_i(\bid)$ denotes the probability that agent $i$ is selected for the task and $r_i(\bid)$ is the reward of agent $i$ under bids $\bid$. Similar to the case of non-revelation mechanisms, we assume that agent $i$ cannot report a time smaller than $t_i$, i.e., $\hat{t}_i \geq t_i$, since solving faster than her true time is infeasible, yielding zero reward.  A mechanism can select multiple agents to perform the task, i.e., $\sum_{i \in [n]} x_i(\bid)$ could be (and decentralization implies that it is) larger than $1$. Since the total available reward is $1$, it must be that $\sum_{i\in[n]} r_i(\bid) \le 1$. 

A revelation mechanism is \emph{incentive compatible}(henceforth, \textbf{IC}) if no agent has an incentive to misreport their type, assuming all other agents report their type honestly. That is, a revelation mechanism is IC if truthtelling is a Nash equilibrium i.e., for all types $\bid=((c_1,t_1),\ldots,(c_n,t_n))$ and for all $b_i$ it holds that $r_i( (c_i,t_i); \bid_{-i}) -  c_i \cdot x_i( (c_i,t_i); \bid_{-i}) \geq r_i( b_i; \bid_{-i}) -  c_i \cdot x_i( b_i; \bid_{-i})$. A mechanism is \emph{individually rational} (henceforth, \textbf{IR}) if every agent gets non-negative utility when reporting her true type, i.e., for all $i\in [n]$, for all types $(c_i,t_i)$, and for all $\bid_{-i}$ it holds that $r_i( (c_i,t_i); \bid_{-i}) -  c_i \cdot x_i( (c_i,t_i); \bid_{-i}) \geq 0$. We assume throughout the paper that revelation mechanisms ignore agents that are ``dominated,'' in the sense of having higher cost and slower time than other agents (since those bids can only arise from dishonest behavior).


\begin{example}[Inverse $k$-price auction]\label{example: inverse k price} Assume that the reported costs satisfy $c_1 \leq \dots \leq c_n$, let $\Mech^{\text{I-$k$ price}}$ be the revelation mechanism that allocates the task to the $k$ agents with the smallest (reported) costs, where $k$ is the largest integer such that $k \cdot c_{k+1} \leq 1$. The reward of agent $i$ is $c_{k+1}$, for all $i = 1, \dots, k$ (i.e., the mechanism is an inverse $k$ price auction). $\Mech^{\text{I-$k$ price}}$ is individually rational ($c_i \leq c_{k+1}$) and IC, since if agent $i$ increases her reported cost her reward will either stay the same or she will be ignored (due to observable dishonest behavior: having a high cost and slow time), and if agent $i$ decreases her reported cost her reward is at most $c_{k+1}$ (or $c_k$ for $i > k$, in which case this reward is not sufficient).
\end{example}

\textbf{Decentralization and Efficiency.} The \textit{decentralization factor} is the maximum number of agents that can be given a reward at least as large as their cost.

\begin{definition}[Decentralization factor $\ks$]\label{dfn: decentralization factor}
For a set of $n$ agents with types $(c_1,t_1), \ldots, (c_n,t_n)$, the decentralization factor $\ks$ is the size of the largest set of agents that can be given a reward at least as large as their cost. Formally, $\ks = \max\{ |S|: S \subseteq [n], \sum_{i \in S} c_i \le 1 \}$.
 
\end{definition}

We use $\ks$ as a benchmark for decentralization since it is precisely the maximum number of agents that the system could reward adequately, in the absence of strategic behavior. As we see in the technical sections, perfect decentralization is not compatible with strategic behavior, and therefore seek approximately decentralized outcomes.
A mechanism's outcome\footnote{For non-revelation mechanisms, an outcome is a Nash equilibrium.} is \textit{$\alpha$-decentralized}, for $\alpha\in(0,1]$, if at least $\alpha \ks$ agents are submitting solutions; a mechanism (revelation or non-revelation) is \textit{$\alpha$-decentralized}, $\alpha\in(0,1]$, if all its outcomes are $\alpha$-decentralized.

We are interested in balancing decentralization with efficiency (fast solutions). In terms of efficiency, a first benchmark could be the fastest time an agent can submit a solution, i.e. $\min_{i \in [n]} t_i$. However, this benchmark might not be even approximately possible when combined with non-trivial decentralization, since the agent that can provide the fastest solution might require the entire reward for herself. 
Therefore, we define efficiency in terms of the fastest solution of an $\alpha$-decentralized outcome.

\begin{definition}[Time guarantee $t_{\alpha}^*$]\label{def: timeGuarantee}
   For a set of $n$ agents with types $(c_1,t_1), (c_2,t_2), \ldots, (c_n,t_n)$, the time guarantee $t_{\alpha}^*$, for $\alpha \in (0,1]$, is the fastest $\alpha$-decentralized set where agents are rewarded at least their cost, $t^*_{\alpha} = \min_{\substack{S: |S|\ge \alpha \ks, \sum_{i\in S} c_i\le 1}} \min_{i \in S} t_i.$
\end{definition}

For example, consider an instance with agent types $(1, \epsilon)$, $(1/2, 1)$ and $(1/2, 1)$. Here, $\ks=2$. $t_{1}^*$ is the fastest time in any set with $\ks=2$ participating agents; in the example, $t_{1}^* = 1$ . Similarly, $t_{1/2}^*$ is the fastest time in any set with $\frac{1}{2} \ks = 1$ participating agents; in the example, $t_{1/2}^* = \epsilon$. In general, as the following proposition shows,  as $\alpha$ becomes smaller, $t^*_{\alpha}$ becomes smaller as well.

\begin{proposition}\label{prop:monotonicityTime}
$t^*_\alpha$ is a monotone non-decreasing function of $\alpha$. That is, for all $\alpha < \alpha'$, $t^*_\alpha \le t^*_{\alpha'}$. 
\end{proposition}

\begin{proof}
    It suffices to show the proposition for $\alpha \in \{ \frac{1}{\ks}, \frac{2}{\ks},\cdots, 1 \}$, where $\alpha \ks$ is an integer.
    Let $S^*_{\alpha + \frac{1}{\ks}} $ be the set that achieves $t^*_{\alpha+\frac{1}{\ks}}$, i.e., $S^*_{\alpha+ \frac{1}{\ks}} \in \min_{\substack{S: |S|\ge (\alpha+\frac{1}{\ks}) \ks, \sum_{i\in S} c_i\le 1}} \min_{i \in S} t_i$.
    Let $S = S^*_{\alpha+\frac{1}{\ks}} \setminus \{\argmax_{i\in S^*_{\alpha+\frac{1}{\ks}}} t_i\}$. The set $S$ is a of size $\alpha \ks$ since $|S| = \alpha \ks$ with cost $\sum_{i\in S}c_i < \sum_{i\in S^*_{\alpha+\frac{1}{\ks}
    }}c_i \le 1 $. Therefore
    $t^*_\alpha \le \min_{i\in S} t_i = \min_{i \in S^*_{\alpha+1}} t_i  = t^*_{\alpha+1}.$
\end{proof}

We say that an outcome with a set $S$ of participating agents is \textit{$\beta$-efficient} if $\min_{i\in S} t_i \le t^*_\beta$; a mechanism (revelation or non-revelation) is \textit{$\beta$-efficient}, $\beta \in (0,1]$, if all its outcomes are $\beta$-efficient. Note, that perhaps counter-intuitively, as $\beta$ becomes smaller, $\beta$-efficiency is a \emph{stronger} guarantee. That is, an $\epsilon$-efficient mechanism, for a small $\epsilon > 0$, has faster solutions than a $1$-efficient mechanism. The following definition and claim will be useful when arguing the efficiency of our mechanisms.

\begin{definition}[$k$-Best Set]
    For agents with types $(c_1,t_1), \ldots, (c_n,t_n)$, such that $c_1\le\ldots\le c_n$, and for a $k \le\ks$ the \emph{$k$-best set} is defined as $S^*_k=\{1,\ldots,k-1\}\cup\left\{\argmin_i\left(t_i:c_i + \sum_{j=1}^{k-1} c_j\le1\right)\right\}$
\end{definition}

The $k$-Best set is not unique, but all $k$-Best sets have the same total cost.

\begin{claim}\label{claim:BestSetEfficient}
    For agents with types $(c_1,t_1),\ldots,(c_n,t_n)$, such that $c_1\le\ldots\le c_n$, with decentralization factor $\ks$, for any $k\le\ks$ the $k$-Best set is $\frac{k}{\ks}$-efficient.  
\end{claim}

\begin{proof}
    For the sake of contradiction, assume that the $k$-best set $S^*_k$ is not $\frac{k}{\ks}$-efficient, i.e. $\min_{i\in S^*_k } t_i > t^*_{\frac{k}{\ks}}$. By \Cref{def: timeGuarantee}, there exist a  set that is $\frac{k}{\ks}$-decentralized and $\frac{k}{\ks}$-efficient, let that be $\hat{S}$. Let $i^*\in\hat{S}$ be the index that $t_{i^*} = t^*_\frac{k}{\ks}$. Since $\hat{S}$ is a feasible set, we have that $c_{i^*} + \sum_{j\in \hat{S} \setminus \{i^*\}} c_j \le 1$. By construction of $S^*_k$, since $i^*\notin \hat{S}$ it must be that $c_{i^*} + \sum_{j=1}^{k-1}c_j > 1$ otherwise it would be in $S^*_k$. That contradicts the feasibility of $\hat{S}$, since $\sum_{j=1}^{k-1}c_j$ is the cost of the cheapest set of size $k-1$. 
\end{proof}

In the absence of strategic behavior, it is possible to simultaneously satisfy $\alpha$-decentralization and $\alpha$-efficiency. 
Our definitions allow the comparison of an $\alpha$-decentralized outcome ($\alpha \ks$ agents participating) with a $\beta$-efficient outcome (the fastest set with $\beta \ks$ participating agents). For $\beta > \alpha$ the monotonicity of the time guarantee,~\Cref{prop:monotonicityTime}, implies $t^*_{\beta} \ge t^*_{\alpha}$.

\paragraph{Useful technical facts.}
Before the technical sections, we state two lemmas that will be useful throughout the paper.

\begin{lemma}
    \label{lemma:upperBoundHarmonic}
For $m$ non-negative real numbers, $0 \le x_1\le x_2 \le \ldots \le x_m$ such that $\sum_{i=1}^m x_i \le 1$, it holds that for all $i \in [m]$: $x_i \le \frac{1}{m-i+1}$.
\end{lemma}

\begin{proof}
    We have that $1 \geq \sum_{k=i}^m  x_k = x_m + x_{m-1} + \dots + x_i \geq x_i + x_i + \dots + x_i = (m-i+1) x_i$. Therefore, $x_i \le \frac{1}{m-i+1}$.
\end{proof}

\begin{lemma}[\cite{friedman2015dynamic}] \label{lemma:sumBounds}
    For natural numbers $b>a>1$,
    \[ \ln(\frac{b}{a-1}) - \frac{1}{2a-2} \le \sum_{j=a}^b \frac{1}{j}  \le \ln(\frac{b}{a-1}).\]
\end{lemma}

\section{Non-Revelation Mechanisms}\label{sec:non-revelation}

In this section, we study non-revelation mechanisms. We first prove a simple, tight bound of $1/2$ for the optimal decentralization (Theorems~\ref{thm: half lower bound non-revelation} and~\ref{thm:equalReward}). This bound can be improved to $1-1/e$ (also in a tight way, under a natural assumption) as $\ks$ grows (Theorems~\ref{thm:harmonicRewards} and~\ref{thm: strong lower bound non revelation}). In~\Cref{subsec: efficiency lower bound} we prove that, unfortunately, non-trivial decentralization is incompatible with efficiency (\Cref{thm:efficiency plus decentralization lower bound non revelation}). Finally, we show how to bypass~\Cref{thm:efficiency plus decentralization lower bound non revelation} in a very strong way: assuming some structure on the agents' types, decentralization and efficiency can be combined optimally, as well as in the case where agents are non-strategic (\Cref{thm: restricted types non revelation}).

\subsection{Warm-up: the Equal Reward mechanism is asymptotically optimal}\label{subsec:warm up non revelation}

We start by proving a simple, \emph{tight} bound for decentralization: in the worst-case, one cannot incentivize (in equilibrium) strictly more than $\ks/2$ agents to participate in a non-revelation mechanism (\Cref{thm: half lower bound non-revelation}). A simple non-revelation mechanism, rewarding all participating agents equally, has, in equilibrium, at least $\ks/2$ participants (\Cref{thm:equalReward}).

\begin{theorem}\label{thm: half lower bound non-revelation}
There is no $\alpha$-decentralized non-revelation mechanism, for any constant $\alpha > \frac{1}{2}$.
\end{theorem}

\begin{proof}
Suppose there is a non-revelation mechanism $\Rule$ that is $\alpha$-decentralized, for some $\alpha>\frac{1}{2}$.

Consider the case of $n=2$ agents whose types are $(1/2,T)$ and $(1/2,T)$. Since $\ks = 2$, and since $\Rule$ is $\alpha$-decentralized for $\alpha>\frac{1}{2}$, it must be that in all equilibria of $\Rule$, both agents submit solutions. First, since both agents submit solutions, and therefore both incur a cost of $1/2$, both must be rewarded at least $1/2$ (and therefore exactly $1/2$, since the total reward is at most $1$). Second, since the agents cannot submit a solution after time $T$, or before $T$ (since their type is $(.,T)$), the equilibrium is such that they both submit a solution at time $T$. Therefore, we can conclude that $\Rule(T,T) = (1/2, 1/2)$.

Now, consider the case of $n=2$ agents whose types are $(2/3,T)$ and $(1/3,T)$. Again, $\ks = 2$, and since $\Rule$ is $\alpha$-decentralized for $\alpha>\frac{1}{2}$, it must again be that in all equilibria of $\Rule$, both agents submit solutions. Agents can only submit a solution at time $T$ (or not at all), and since $\Rule(T,T) = (1/2, 1/2)$, the first agent, whose cost is $2/3 > 1/2$, prefers to not submit a solution if the other agent submits one. Therefore $\Rule$ can only have one agent submitting a solution in equilibrium; a contradiction.
\end{proof}

\begin{theorem}[Equal Reward Rule]\label{thm:equalReward}
Let $\Rule^{\text{eq}}$ be the non-revelation mechanism that rewards all participating agents equally. That is, $\Rule^{\text{eq}}_i(a_1, \dots, a_n) = 1/\ell$ for all $i$ such that $a_i \neq \bot$, where $\ell = |\{ i \in [n]: a_i \neq \bot \}|$. Then, $\Rule^{\text{eq}}$ is $\frac{1}{2}$-decentralized.
\end{theorem}

\begin{proof}
First, we prove that all pure Nash equilibria of $\Rule^{\text{eq}}$ have the same number of agents. Then, we show how, given an arbitrary instance, we can construct a pure Nash equilibrium\footnote{Thus, also showing that pure Nash equilibria exist for $\Rule^{\text{eq}}$.} of size at least $\ks/2$. Combined, the two statements conclude the proof of the theorem.

Towards the first statement, assume that $S_A \subseteq [n]$ and $S_B \subseteq [n]$ are the agents that submit solutions in two equilibria of $\Rule^{\text{eq}}$, such that $|S_A| + 1  \le |S_B|$. By definition, $\forall i\in S_A$ $c_i \le 1/|S_A|$ and $\forall j \in S_B$ $c_j \le 1/|S_B|$. There exists an agent $i^*$ such that $i^* \in S_B \setminus S_A$, and since $i^* \in S_B$ we have that $c_{i^*} \le 1/|S_B| \le 1/(|S_A|+1)$. Therefore, only the set $S_A$ submitting a solution cannot be an equilibrium, since $i^*$ prefers submitting a solution to not submitting a solution;\footnote{Recall agents prefer submitting a solution over not submitting when their expected utilities are equal.} a contradiction.

Towards the second statement, we construct an equilibrium as follows. First, without loss of generality, rename the agents such that $c_1 \le \ldots \le c_n$. Notice that $\ks$ is the largest index $k$ such that $\sum_{i=1}^k c_i \leq 1$. Let $\ell$ be the largest index such that $c_{\ell} \leq \frac{1}{\ell}$ and $c_{\ell + 1} > \frac{1}{\ell + 1}$ (and $\ell = n$ if  $c_i \leq 1/i$ for all $i$). By construction, the set of agents $\{ 1, 2, \dots, \ell \}$ submitting a solution is an equilibrium of $\Rule^{\text{eq}}$. Finally,
$ 1 \ge \sum_{i=1}^{\ks} c_i = \sum_{i=1}^{\ell} c_i +\sum_{i=\ell+1}^{\ks} c_i \ge \sum_{i=\ell+1}^{\ks} c_i  > \frac{\ks-\ell}{\ell+1}.$

Re-arranging we have that $\ell > \frac{\ks-1}{2}$. When $\ks$ is an odd number, this implies $\ell \geq \lceil \frac{\ks}{2} \rceil$; when $\ks$ is an even number, this implies that $\ell \geq \frac{\ks}{2}$. So, overall we have $\ell \geq \frac{\ks}{2}$.
\end{proof}

\subsection{The Harmonic Rule: better decentralization for large $\ks$}\label{subsec: harmonic}

Theorems~\ref{thm: half lower bound non-revelation} and~\ref{thm:equalReward} give a tight bound on the decentralization guarantees of non-revelation mechanisms. However, the instances where this bound is witnessed involve small values of $\ks$: $\ks=2$. We prove that, when $\ks$ is large, better guarantees are possible.

\begin{theorem}[Harmonic Rule] \label{thm:harmonicRewards}
      Let $\Rule^{\text{Harm}}$ be the non-revelation mechanism that rewards the $i$-th fastest participating agent $\frac{1}{m + i - 1}$, where $m$ is the smallest non-negative integer for which the overall reward is at most $1$. That is, if $\ell = |\{ i \in [n]: a_i \neq \bot \}|$ is the number of agents participating, and agent $j$ submitted the $i$-th fastest solution, the reward of $j$ is $\frac{1}{m+i-1}$, where $m = \argmin_{a \in \mathbb{N}} \{ \sum_{x=a}^{a+\ell-1} \frac{1}{x} \leq 1 \}$. Then, $\Rule^{\text{Harm}}$ is $\left(1-\frac{1}{e} - \frac{8}{\ks}\right)$-decentralized.
\end{theorem}

\begin{proof}
     Assume w.l.o.g. that $c_{i} \le c_{i+1}$ for all $i$. Let $m_{z} = \argmin_{a \in \mathbb{N}} \{ \sum_{x=a}^{a+z-1} \frac{1}{x} \leq 1 \}$.

     First, we prove that a pure Nash equilibrium exists via induction on $n$. For $n=1$ the statement is trivial. Assume that a pure Nash equilibrium exists for agents with costs $c_{1} \leq \dots \leq c_{n-1}$, and let $S^*$ be the set of agents participating in this equilibrium, where $\ell = |S^*|$. Let $rank(i,S)$ be the ranking of agent $i$ among a subset $S$ of agents, where $rank(i,S) = 1$ means that $t_i$ is the smallest among agents in $S$, and $rank(i,S) = |S|$ means that $t_i$ is the largest among agents in $S$. For all $i=1,\dots,n-1$ such that $i \notin S^*$, we know that $c_i$ is strictly larger than the reward they get when they participate in $S^* \cup \{ i \}$, i.e. $c_i > \frac{1}{m_{\ell+1} + rank(i, S^* \cup \{ i \}) - 1}$. Consider an instance with costs $c_{1} \leq \dots \leq c_n$. If $S^*$ is not an equilibrium, it must be that agent $n$ wants to join, i.e. $c_n \leq \frac{1}{m_{\ell+1}}$. If $S^* \cup \{ n \}$ is not an equilibrium, it must be an agent $S^*$ is not adequately rewarded (i.e. it is not the case that an additional agent wants to join); let $i^*$ be the highest cost agent in $S^*$ that is not adequately rewarded. Finally, $S^* \cup \{ n \} \setminus \{ i^* \}$ must be an equilibrium: (i) agent $n$ gets reward $\frac{1}{m_{\ell}} \geq \frac{1}{m_{\ell+1}} \geq c_n$, (ii) agent $i \in S^* \setminus \{ i^* \}$ either gets reward at least her reward in $S^*$ if her rank didn't change or, otherwise, her rank did change compared to $rank(i,S^*)$, but her cost is at most her reward in $S^* \cup \{ n \}$ (since $i^*$ has the highest cost and unhappy agent) (iii) agent $i \notin S^* \cup \{ n \} \setminus \{ i^* \}$ has a cost strictly larger than her reward if she joins, since $rank(i, S^* \cup \{ i \}) \leq rank(i, S^* \cup \{ n, i \} \setminus \{ i^* \})$.

     Next, we prove the decentralization bound. Consider a pure Nash equilibrium of $\Rule^{\text{Harm}}$ with $\ell$ agents participating. When an additional agent submits a solution, her reward is at least $\frac{1}{m_{\ell+1}}$; therefore, for all $i$ such that $a_i = \bot$, $c_i > \frac{1}{m_{\ell+1}}$. Let $i^*$ be the agent with the smallest cost among those not participating, noting that $i^* \leq (1-1/e) \ks - 8$, since we sorted in increasing cost, and since, otherwise, the theorem follows immediately. 
     Since $\sum_{i=1}^\ks c_i \leq 1$,~\Cref{lemma:upperBoundHarmonic} implies that $\frac{1}{\ks - {i^*} + 1} \geq c_{i^*} > \frac{1}{m_{\ell+1}}$. Re-arranging we have 
     \begin{equation}\label{eq: m ell}
     m_{\ell+1} > \ks - i^* + 1  \geq \frac{\ks}{e} + 9.
     \end{equation}

     We know that $m_{\ell+1}$ is the smallest non-negative integer such that $\sum_{x=m_{\ell+1}}^{m_{\ell+1}+\ell} \frac{1}{x} \leq 1$. Note that this does not imply that $\sum_{x=m_{\ell+1}}^{m_{\ell+1}+\ell + 1} \frac{1}{x} > 1$; for example, for $\ell+1 = 3$, $1/2 + 1/3 + 1/4 > 1$, $1/3 + 1/4 + 1/5 < 1$ (so, $m_3 = 3$), but also $1/3 + 1/4 + 1/5  + 1/6 < 1$ (that is, $m_4=3$). However, it is true that $\sum_{x=m_{\ell+1}-1}^{m_{\ell+1}+\ell-1} \frac{1}{x} > 1$. The next claim shows that we can always add a constant number of terms after $\frac{1}{m_{\ell+1}+\ell}$ so that $\sum_{x=m_{\ell+1}}^{m_{\ell+1}+\ell + c} \frac{1}{x} > 1$; this, in turn, will give us an upper bound on $m_{\ell+1}$ that we can combine with~\Cref{eq: m ell}.

     \begin{claim}\label{claim: m ell claim}
         For all $m_{\ell+1} \geq 2$, $\sum_{x=m_{\ell+1}}^{m_{\ell+1}+\ell + 20} \frac{1}{x} > 1$.
     \end{claim}
\begin{proof}
     Since $\sum_{x=m_{\ell+1}-1}^{m_{\ell+1}+\ell-1} \frac{1}{x} > 1$, it suffices to show that $\frac{1}{m_{\ell+1}-1} < \sum_{x = m_{\ell+1}+\ell}^{m_{\ell+1}+\ell + c} \frac{1}{x}$.

     We know that $1 \geq \sum_{x=m_{\ell+1}}^{m_{\ell+1}+\ell} \frac{1}{x} \geq \ln(m_{\ell+1}+\ell) - \ln(m_{\ell+1}-1) - \frac{1}{2m_{\ell+1} - 2} \geq \ln(m_{\ell+1}+\ell) - \ln(m_{\ell+1}-1) - \frac{1}{2}$, where the second inequality uses~\Cref{lemma:sumBounds} and the third inequality uses the fact that $m_{\ell+1} \geq 2$. Re-arranging we have $\frac{m_{\ell+1}+\ell}{m_{\ell+1}-1} \leq e^{1.5}$, or $\frac{1}{m_{\ell+1}} \leq \frac{e^{1.5}}{m_{\ell+1}+\ell} \leq \frac{5}{m_{\ell+1}+\ell}$. Therefore, it suffices to show that $\sum_{x = m_{\ell+1}+\ell}^{m_{\ell+1}+\ell + c} \frac{1}{x} \geq \frac{5}{m_{\ell+1}+\ell}$, or $\sum_{x = m_{\ell+1}+\ell+1}^{m_{\ell+1}+\ell + c} \frac{1}{x} \geq \frac{4}{m_{\ell+1}+\ell}$.

     Using~\Cref{lemma:sumBounds} we have that  $\sum_{x = m_{\ell+1}+\ell+1}^{m_{\ell+1}+\ell + c} \frac{1}{x} \geq \ln\left( \frac{m_{\ell+1}+\ell + c}{m_{\ell+1}+\ell} \right) - \frac{1}{2(m_{\ell+1}+\ell)}$. Re-arranging, and replacing $z = m_{\ell+1}+\ell$ (to simplify notation) and plugging in $c=20$, it suffices to show that
     \[
     \ln\left( 1 + \frac{c}{z} \right) = \ln\left( 1 + \frac{20}{z} \right) \geq \frac{4.5}{z}.
     \]
     One can confirm that this expression is true for $z = 2, 3, 4, 5$ by doing simple calculations. For $z \geq 5$ we use the lower bound $\ln(1+x) \geq \frac{x}{1+x}$ to get $\ln\left( 1 + \frac{20}{z} \right) \geq \frac{20}{20 + z}$, which is at least $\frac{4.5}{z}$ for all $z \geq 6$. This confirms the proof of~\Cref{claim: m ell claim}.
     \end{proof}
    Using~\Cref{claim: m ell claim} and~\Cref{lemma:sumBounds} we have $1 < \sum_{x=m_{\ell+1}}^{m_{\ell+1}+\ell + 20} \frac{1}{x} \leq \ln \left( \frac{m_{\ell+1}+\ell + 20}{m_{\ell+1}-1} \right)$, which implies $m_{\ell+1} \leq \frac{\ell + 20 + e}{e-1}$. Combining with~\Cref{eq: m ell} we have $\frac{\ell + 20 + e}{e-1} > \frac{\ks}{e} + 8$, or $\ell \geq \left( 1 - \frac{1}{e} \right)\ks + 9(e-1) - (20+e) \geq \left( 1 - \frac{1}{e} \right)\ks - 8$, concluding the proof of~\Cref{thm:harmonicRewards}.
    \end{proof}

The guarantee in~\Cref{thm:harmonicRewards} is also (asymptotically) tight, under the following condition. We say that a non-revelation mechanism is \textbf{reward-monotone} if the $i$-th smallest reward when $\ell$ agents submit solutions is at least the $i$-th smallest reward when $\ell+1$ agents submit solutions.

\begin{theorem}\label{thm: strong lower bound non revelation}
There is no $\alpha$-decentralized and reward-monotone non-revelation mechanism, for any $\alpha > 1 - \left( e^{1+\frac{e^2}{2 \ks}} \right)^{-1}$.
\end{theorem}

\begin{proof}
    Let $\Rule$ be a non-revelation and reward-monotone mechanism that is $\alpha$-decentralized. Fix $\ks$; we will consider instances with $n = \ks$ agents. Let $\ell^* = \lceil \alpha \ks \rceil$ be the number of agents that the mechanism needs to incentivize to participate in order to satisfy $\alpha$-decentralization. We will consider the behavior of $\Rule$ when it receives solutions only at time $T$, noting that, similar to the proof of~\Cref{thm: half lower bound non-revelation}, submitting a solution at time $T$ and not submitting a solution are the only options for agents whose types are $(c,T)$. 
    Specifically, let $r^{(z)}_i$ be the $i$-th smallest reward when $z$ agents submit solutions (at time $T$), where $z = \ell^*, \ell^* + 1, \dots, \ks$.

    First, notice that, for some $z$, the smallest reward must be at least $1/\ks$, otherwise when all agents have costs $1/\ks$, the agent receiving the smallest reward wants to deviate. Similarly, the second smallest reward must be at least $1/(\ks-1)$, and so on. This is formalized in the following claim.

    \begin{claim}\label{claim: reward monotonicity}
        For all $i = 1, \dots, \ell^*$, $\min_z r^{(z)}_i \geq \frac{1}{\ks + 1 - i}$.
    \end{claim}

    \begin{proof}
            Assume that for some $i$ $\min_z r^{(z)}_i  = \frac{1}{\ks + 1 - i} - \delta$, for some $\delta > 0$. Consider an equilibrium for the case where $\ks + 1 - i$ agents have types $(\frac{1}{\ks + 1 - i} - \frac{\delta}{n},T)$ and $i-1$ agents have types $(\frac{\delta}{n},T)$; the number of agents participating in this equilibrium is at least $\ell^* \geq i$. The $i$-th smallest reward is at most $\frac{1}{\ks + 1 - i} - \delta < \frac{1}{\ks + 1 - i} - \frac{\delta}{n}$, and therefore one of the agents with cost $\frac{1}{\ks + 1 - i} - \frac{\delta}{n}$ (the $i$-th smallest cost) cannot be rewarded adequately, contradicting the equilibrium condition.
    \end{proof}

    Since $\Rule$ is reward-monotone,~\Cref{claim: reward monotonicity} implies that $r^{(\ell^*)}_i \geq \frac{1}{\ks + 1 - i}$. Since $r^{(\ell^*)}_1, \dots, r^{(\ell^*)}_{\ell^*}$ must be a feasible set of rewards, we have 
    \begin{equation}\label{ineq: harmonic bound 1}
          1 \geq \sum_{i=1}^{\ell^*} r^{(\ell^*)}_i \geq \sum_{i=1}^{\ell^*} \frac{1}{\ks + 1 - i} = \sum_{i = \ks + 1 - \ell^*}^{\ks} \frac{1}{i}.  
    \end{equation}
    Using the fact that $\ln(b)+1 \geq H_b \geq \ln(b)$, we can lower bound the RHS of~\Cref{ineq: harmonic bound 1} by $\ln(\ks) - \ln(\ks - \ell^*) - 1$; re-arranging and simplifying we have $\ks - \ell^* \geq \ks/e^2$. Using this bound we can apply~\Cref{lemma:sumBounds} to lower bound the RHS of~\Cref{ineq: harmonic bound 1} as $\ln(\ks) - \ln(\ks - \ell^*) - \frac{1}{2(\ks - \ell^*)} \geq \ln(\ks) - \ln(\ks - \ell^*) - \frac{e^2}{2 \ks}$. Re-arranging and simplifying we have $\ell^* \leq \left( 1 - \frac{1}{e^{1+\frac{e^2}{2 \ks}}} \right) \ks$.
\end{proof}

\subsection{Decentralization, efficiency, and non-overlapping types}\label{subsec: efficiency lower bound}

So far we have focused on decentralization and ignored efficiency. Optimal efficiency is trivial, by implementing the non-revelation mechanism that only rewards the fastest solution; however, this solution achieves trivial decentralization of $1/\ks$. Our next result shows that non-trivial decentralization is incompatible with efficiency guarantees. Recall that $1$-efficiency is the \emph{weakest} efficiency guarantee, which asks that one of the submitted solutions (in equilibrium) is as fast as the fastest solution in the best set with $\ks$ submitted solutions (i.e. a set with optimal decentralization).

\begin{theorem}\label{thm:efficiency plus decentralization lower bound non revelation}
There is no $\alpha$-decentralized and $1$-efficient non-revelation mechanism, for any $\alpha > \frac{1}{\ks}$.
\end{theorem}

\begin{proof}
    Suppose that there is a non-revelation mechanism $\Rule$ that is simultaneously $\alpha$-decentralized and $1$-efficient, for some $\alpha> \frac{1}{\ks}$. 

    Consider the rewards $\Rule$ produces when there are $\ell \geq 1$ solutions at time $T$ and one solution at time $t < T$; slightly overloading notation, let $\Rule_1(\ell, t)$ be the reward of the agent that submits the solution at time $t$, and let $r_{max}$ be the maximum reward this agent can receive in an equilibrium where $\ell \geq 1$ other agents submit solutions at time $T$.
    If such an equilibrium does not exist, then $\alpha$-decentralization is violated, since we could pick an instance with $\ks$ large enough so that $\alpha \ks > 1$.
    
    Since $r_{max}$ is the reward of one agent in an equilibrium with multiple participating agents, it must be that $r_{max} < 1$. Let $r_{max} = 1 - \epsilon^*$. Next, consider the case where there are $\ks - 1$ agents with type $(\epsilon,T)$ and one agent with type $(1-(\ks-1)\epsilon, \delta)$, for a small $\delta > 0$. $1$-efficiency implies that the fast agent must be participating in all equilibria. We can pick $\ks$ large enough so that $\alpha$-decentralization will imply that all equilibria have at least two participating agents. However, by picking $\epsilon$ such that $1-\epsilon^* < 1-(\ks-1)\epsilon$, we have that $r_{max}$, the maximum possible reward in an equilibrium, is not an adequate reward for the fast agent. Therefore, no equilibrium where this agent and at least one more agent participates exists; a contradiction.
\end{proof}

\Cref{thm:efficiency plus decentralization lower bound non revelation} shows that decentralization and efficiency are not compatible when agents are strategic. Intuitively, what goes wrong is that the fastest agent \emph{must} be incentivized to participate, but there is no way to reward her adequately while also rewarding a second agent.  To bypass this negative result, here we assume that agents' types satisfy the following structure on types, which we call \textbf{non-overlapping types}. Consider $m$ profile ``buckets,'' $B = \{ ([\underline{t}_1, \bar{t}_1], \hat{c}_1),\ldots, ([\underline{t}_m, \bar{t}_m], \hat{c}_m)\}$, characterized by a computation time interval and an upper bound on computation costs. The costs are distinct and monotone: $\hat{c}_1,< \ldots <\hat{c}_m$. Furthermore, there are no overlapping time intervals, that is, for all $j<m$, $\bar{t}_{j+1}<\underline{t}_{j}$. Overloading notation, we say $t\in B_j$ if $t\in[\underline{t}_j, \bar{t}_j]$. For an agent with computation time $t$, let $b(t) = \{ j:t\in B_j \}$ ($b(\bot)=\bot$) be the function that returns her bucket and  $c(t) = \hat{c}_{b(t)}$ ($c(\bot)=\bot$ be the function that returns her cost. In this model, a non-revelation mechanism can leverage the fact that the time of submission is correlated with the cost in order to achieve decentralization and efficiency simultaneously.

Our non-revelation mechanism, $\Rule^{\text{Best-Set}}$, takes as input a parameter $k \leq \ks$ (the target size of the outcome) and works as follows. If the number of solutions is $m' < k$, $\Rule^{\text{Best-Set}}$ rewards the $z$ cheapest agents exactly their cost (for $z$ as large as possible). Otherwise, $\Rule^{\text{Best-Set}}$ finds the $k$-best set of $(c(a_1),a_1),\ldots,(c(a_n),a_n)$ with the fastest $k-1$ agents with the smallest cost.
Formally, to find the $k$-best set, first, $\Rule^{\text{Best-Set}}$ partitions the agents into their corresponding buckets $\hat{B}_j = \{i: b(a_i)=j\}$ ($\hat{B}_\bot = \{i: b(a_i) = \bot \} $). Then, it picks the $k-1$ (fastest) agents with the smallest cost. 
For $\hat{B}_1,\ldots,\hat{B}_m,\hat{B}_\bot$, let $z^*=\argmin_{z}\{\sum_{j=1}^z|\hat{B}_j|> k-1\}$ be the index of the first bucket that cannot be rewarded in its entirety, and let $S^{(1)}=\cup_{j=1}^{z^*-1}\hat{B}_j$ and $S^{(2)}\subset \hat{B}_{z^*}$ be the fastest agents such that  $|S^{(1)}\cup S^{(2)}|=k-1$. Finally, let $\ell=\argmin_i\{t_i:\sum_{j\in S^{(1)}\cup S^{(2)}}c_j \le 1\}$. The reward of agent $i$ is $c(a_i)\cdot\mathds{1}\{ i\in S^*_k\}$, where $S^*_k = S^{(1)}\cup S^{(2)}\cup\{\ell\}$.

\begin{theorem}[Best-Set Rule]\label{thm: restricted types non revelation}
When types are non-overlapping, $\Rule^{\text{Best-Set}}$, on input $k\le\ks$, is $\frac{k}{\ks} 
$-decentralized and $\frac{k}{\ks}$-efficient.

\end{theorem}

\begin{proof}

Let $(c_1,t_1),\ldots, (c_n,t_n)$ be the underlying instance. It is easy to see that no agent has an incentive to delay their submission, since by delaying she can only be miscategorized into a slower bucket and receive a smaller reward. Therefore, for all $i\in[n]$, $a_i \in \{t_i, \bot\}$, and thus we only need to consider as actions ``participate'' or ``not participate''.  

First, we show that a pure Nash equilibrium exists. Let $\hat{S}$ be the $k$-best set of $\Rule^{\text{Best-Set}}$ when all agents participate. Let $\ell$ be the index of the fastest agent in $\hat{S}$. When only the agents in $\hat{S}$ participate, this is an equilibrium. To see this, consider any agent $i\in\hat{S}$; she has no incentive to deviate since she is rewarded according to her profile and cannot submit a solution in a faster (and more expensive) bucket. Now, consider any agent $i\notin \hat{S}$. There are 4 cases that we need to examine about the bucket $b(t_i)$. Let $j^*$ be the bucket of the fastest agent in $\hat{S}\setminus\{\ell\}$(fastest among the slow solutions).  If $b(t_i)>b(t_\ell)$, it must be that $c(t_i) + \sum_{\hat{S}\setminus\{\ell\}}c_j>1$, hence the agent prefers to not participate. If $b(t_i)=j^*$ it must be that $t_i > \min_{z\in \hat{S}\setminus\{\ell\}} t_z$ (i.e., not the fastest in the bucket) since $i\notin \hat{S}$. Since agents cannot submit faster solutions, there is no way to be included in the set, thus she does not want to participate. If $b(t_i)=b(t_\ell)$, the same argument applies since $i$ is in the same bucket as an agent in $\hat{S}$ but not the fastest in the bucket. Finally, for $i\notin \hat{S}$ such that  $j^*<b(t_i)<b(t_\ell)$ (not in a rewarding bucket), notice that if $i$ participates, $\hat{S}$ will not change since $i$ is neither in the cheapest bucket nor is the fastest agent available.

We will show that the above equilibrium is unique. First, it is easy to see that there is no equilibrium of size $k+1$, since at most $k$ agents get rewards. Next, for any set of submitted solutions, every one of the $k-1$ ``cheapest and fastest'' agents is always incentivized to also submit a solution, therefore all equilibria must include all of the $k-1$ ``cheapest and fastest'' agents. Therefore, we cannot have an equilibrium of size $k-1$ either, since agent $\ell = \argmin_i\{t_i:\sum_{j\in S^{(1)}\cup S^{(2)}}c_j \le 1\}$ prefers to join. Finally, in a set of solutions with the $k-1$ cheapest agents and another agent that is not $\ell$ cannot be an equilibrium, since agent $\ell$ prefers to join.

$\frac{k}{\ks}$-decentralization is immediately implied. To conclude the proof, we need to show that the rule is also $\frac{k}{\ks}$-efficient. That is an immediate implication of~\Cref{claim:BestSetEfficient}.
\end{proof}

\section{Revelation Mechanisms}\label{sec: revelation}

In this section we study revelation mechanisms. We first prove an analogue to~\Cref{thm: half lower bound non-revelation}.

\begin{theorem}\label{thm: revelation decentralization lower bound}
No IC and IR revelation mechanism is $\alpha$-decentralized, for any constant $\alpha > 1/2$.
\end{theorem}

\begin{proof}
Suppose that there is a revelation mechanism $\Mech$ that achieves $(1/2 + \epsilon)$-decentralization, for some $\epsilon > 0$. Consider an instance with $n=2$ agents and decentralization factor $\ks=2$: let $\mathcal{I}= ( (c_1,t_1), (c_2,t_2) )$, where $c_1+c_2 < 1$. Now, consider the following two instances  $\mathcal{I}_1 = ( (\widehat{c_1},t_1), (c_2,t_2) )$ and 
$\mathcal{I}_2=( (c_1,t_1), (\widehat{c_2},t_2) )$ such that $\widehat{c_1}+c_2=1$ and $c_1+\widehat{c_2}=1$, respectively. By assumption, $\Mech$ must incentivize both agents to submit solutions in all instances, i.e. the sum of allocations must be equal to $2$. By individual rationality, it must then be that $r^{\Mech}( (\widehat{c_1},t_1), (c_2,t_2) ) = (\widehat{c_1},c_2)$ and $r^{\Mech}( (\widehat{c_1},t_1), (c_2,t_2) ) = (\widehat{c_1},c_2)$. 
Now, if $r_1^\Mech( (c_1,t_1), (c_2,t_2) ) < \widehat{c_1}$, agent $1$ has an incentive to deviate to reporting $(\widehat{c_1},t_1)$. Similarly, if $r_2^\Mech( (c_1,t_1), (c_2,t_2) ) < \widehat{c_2}$ agent 2 is incentivized to report $(\widehat{c_2}, t_2)$. Combining the above inequalities, we have that $r_1^\Mech( (c_1,t_1), (c_2,t_2) ) \ge \widehat{c_1} $ and $r_2^\Mech( (c_1,t_1), (c_2,t_2) ) \ge \widehat{c_2} $. However, $c_1<\widehat{c_1}$, since $c_1+c_2 < 1$, and $\widehat{c_1}+c_2= 1$. Thus, $r_1^\Mech( (c_1,t_1), (c_2,t_2) ) + r_2^\Mech( (c_1,t_1), (c_2,t_2) ) \ge \widehat{c_1} + \widehat{c_2} > c_1+\widehat{c_2} =1$; a contradiction.
\end{proof}

\begin{claim}\label{thm: 1/2 upper bound revelation}
The inverse $k$-price auction of~\Cref{example: inverse k price} is $1/2$ decentralized.
\end{claim}

\begin{proof}
    Assume that reported costs satisfy $c_1 \leq c_2 \dots \leq c_n$, and let $k$ be the largest integer such that $k \cdot c_{k+1} \leq 1$. We know that, for $\ks$, $\sum_{i=1}^\ks c_i \leq 1$, therefore,~\Cref{lemma:upperBoundHarmonic} then implies that $c_i \leq \frac{1}{\ks - i + 1}$ for all $i \leq \ks$. By the choice of $k$ we have $1 < (k+1) \cdot c_{k+2} \leq \frac{k+1}{\ks - (k+2) + 1}$. Re-arranging we have $k \geq \frac{\ks}{2} - 1$.
\end{proof}

The fact that the inverse $k$-price auction of~\Cref{example: inverse k price} is, in fact, $1/2$ decentralized implies that~\Cref{thm: revelation decentralization lower bound} is tight. Similar to~\Cref{thm: half lower bound non-revelation}, the bound of $1/2$ can be improved. However, unlike non-revelation mechanisms, optimal efficiency is compatible with optimal decentralization.

Our next result shows a \emph{revelation gap}:~\Cref{thm:efficiency plus decentralization lower bound non revelation} is not true for revelation mechanisms. We give a mechanism parameterized by $k$ that is incentive compatible (IC) and achieves $\frac{\min\{ k, \ks\} -1}{\ks}$-decentralization and $\frac{\min\{ k, \ks\}}{\ks}$-efficiency. Our mechanism, Inverse Generalized Second Price (I-GSP), denoted by $\Mech^{\text{I-GSP}}$, takes as input a parameter $k$ and works as follows. To simplify the description of the mechanism, assume for now that $k \leq \ks.$ Without loss of generality, reported costs satisfy $c_1 \leq \dots \leq c_n$. If there is an inversion, i.e., there is an $i>j$ such that $t_i<t_j$, the mechanism does not allocate anything. Otherwise, $\Mech^{\text{I-GSP}}$ allocates the task to the $k-2$ agents with the smallest cost, as well as the agent $\ell$ with the highest cost possible to maintain feasibility, i.e., $\ell = \argmax_{z} c_z + \sum_{i=2}^{k-1} c_i \leq 1$. Agent $\ell$ gets reward $1-\sum_{i=2}^{k-1} c_i$; for $i=1,\dots,k-2$, the reward of agent $i$ is $c_{i+1}$. If $k > \ks$, then $\sum_{i=1}^{k} c_i > 1$, so the mechanism follows the procedure above for $\ks$ instead of $k$, where $\ks = \argmax_\ell\{\sum_{i=1}^\ell c_i\le 1\}$ of the reported costs. 

\begin{theorem}[Inverse Generalized Second Price (I-GSP)]\label{thm: inverse GSP}
$\Mech^{\text{I-GSP}}$ is IC, IR, and, on input $k$, is $\frac{\min\{ k, \ks\} -1}{\ks}$-decentralized and $\frac{\min\{ k, \ks\}}{\ks}$-efficient revelation mechanism.
\end{theorem}

\begin{proof}
 First, we prove that $\Mech^{\text{I-GSP}}$ is individually rational. It is easy to see that for agents $i=1,\dots,k-2$ the reward $r_i = c_{i+1}\ge c_i.$ Therefore, it suffices to show that $c_{\ell} \leq r_{\ell} = 1-\sum_{i=2}^{k-1} c_i$, which is true by construction since $c_{\ell} + \sum_{i=2}^{k-1} c_i \leq 1$. 


 Second, we prove that $\Mech^{\text{I-GSP}}$ is IC. Consider agent $i$, for $i=1, \dots, k-2$. Recall no agent can under-report her time  ($\hat{t}_i<t_i$) without detection, since she cannot deliver the solution faster. Increasing her time beyond $t_{i-1}$ will cause a detectable inversion with zero utility. Regarding the cost, reporting a cost $\hat{c}_i>c_{i+1}$ either causes an inversion or is detected when paired with  $\hat{t}_i<t_{i+1}<t_i$; both result in zero utility. Any  $\hat{c}_i \in [c_i,c_{i+1})$ does not change her allocation or reward. Decreasing her reported cost can only reduce her reward. For agent $k-1$, reducing her reported cost below $c_{k-2}$ can only result in a reward of at most $c_{k-2}\le c_{k-1}$; increasing her reward might change the choice of $\ell$ but will not give her any reward, since $\sum_{i=2}^{k-1} c_i + c_{\ell} \leq 1$ (and agent $k$ cannot increase her cost above $c_{\ell},$ without causing an inversion). Agents $i$, for $i = k+1, \dots, \ell -1$ do not benefit from decreasing their cost and cannot increase their cost above $c_{\ell}$. Agent $\ell$ cannot improve her reward. Finally, agents $i \geq \ell + 1$ can get a reward of $1 - \sum_{i=2}^{k-1} c_i < c_{\ell + 1}$ by lowering their cost.

 Third, $\sum_{i=1}^n x_i(t_1,\dots,t_n) = k-1$, or $\ks - 1$ if $k > \ks$, therefore $\Mech^{\text{I-GSP}}$ is $\frac{ \min\{ k, \ks \} -1}{\ks}$-decentralized. Fourth, by definition, agent $\ell$ is the fastest agent, among all sets with $\min\{ k, \ks \}$ participants; therefore, $\Mech^{\text{I-GSP}}$ is $\frac{\min\{ k, \ks \}}{\ks}$-efficient. 
\end{proof}

\section{Experiments}\label{sec: experiments}
In this section, we evaluate the performance of our proposed mechanisms using synthetic data in two sets of experiments. In the first set, we compare the decentralization factor of the equal reward mechanism, $R^{\text{eq}}$, and the harmonic mechanism $R^{\text{harm}}$. In the second set, we compare the fastest solution submitted under the non-revelation harmonic mechanism $R^{\text{harm}}$ and the revelation mechanism inverse generalized second price $\Mech^{\text{I-GSP}}$, under the same number of participants.  

In the first set of experiments, we sample the agents' costs from two different probability distributions: Uniform $(0,1)$ and Exponential with the parameter $\lambda=1$ (and normalized to $(0,1)$). These distributions capture different scenarios about the type of participating agents. In the Uniform case, all agents are assumed to be identical. The Exponential distribution represents the scenario when the majority of the agents have low costs (e.g., normal computers) and only a few agents have high costs (e.g., supercomputers).

We perform simulations for instances with $n$ participating agents, where $n\in[1,1000]$. For each instance, we compute the decentralization factor $\ks$, which represents the maximal set of participating agents that can be adequately rewarded. We then compute the number of agents participating in equilibrium under each mechanism, and the corresponding decentralization factor, which is the ratio of participating agents to $\ks$. The experiments are repeated 500 times for each $n$, and the average values are shown. In \Cref{fig: ParticipationU} and \Cref{fig:ParticipationExp} show the number of participants in equilibrium for each mechanism (orange for $R^{\text{harm}}$ and green for $R^{Eq}$) alongside $\ks$ (blue). In \Cref{fig:DecentrU,fig:DecentrExp}, we plot the decentralization factor of $R^{\text{harm}}$ (orange) and $R^{Eq}$ (green). We compare the decentralization factors with the theoretical upper (purple) and lower bounds (red). The results show that, for random instances, both algorithms vastly outperform the theoretical upper bound.

\begin{figure}[ht]
\centering
\begin{subfigure}{.5\textwidth}
  \centering
  \includegraphics[width=0.9\linewidth]{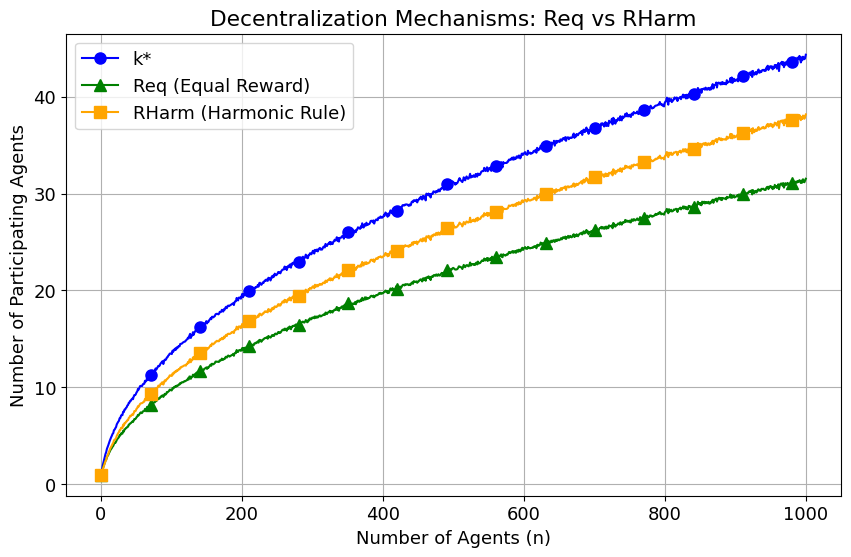}
  \caption{Number of participants}
  \label{fig: ParticipationU}
\end{subfigure}%
\begin{subfigure}{.5\textwidth}
  \centering
  \includegraphics[width=0.9\linewidth]{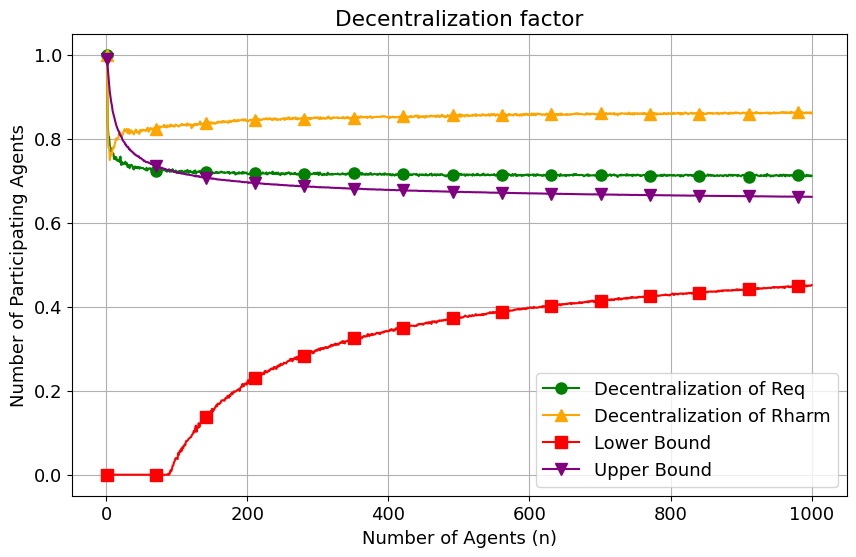}
  \caption{Decentralization factor}
  \label{fig:DecentrU}
\end{subfigure}
\caption{Samples are drawn from Uniform $(0,1)$ distribution}
\label{fig:Uniform}
\end{figure}

\begin{figure}[ht]
\centering
\begin{subfigure}{.5\textwidth}
  \centering
  \includegraphics[width=0.8\linewidth]{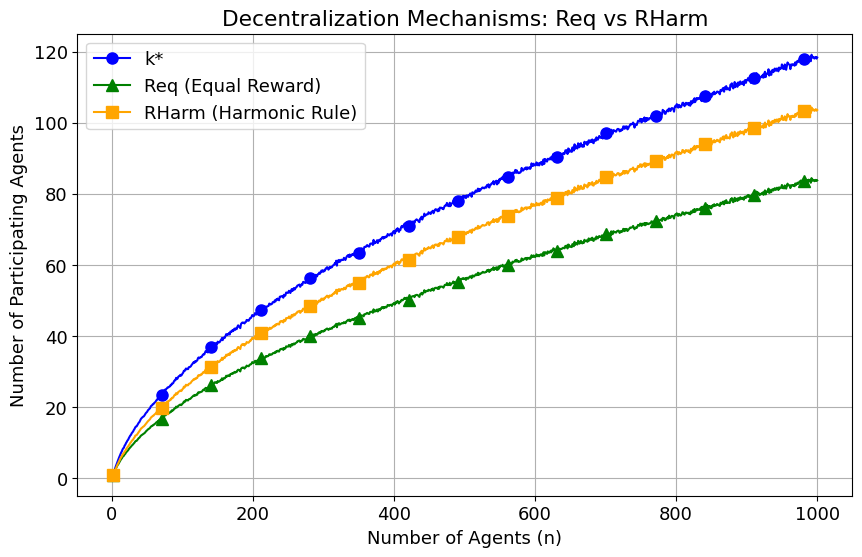}
  \caption{Number of participants}
  \label{fig:ParticipationExp}
\end{subfigure}%
\begin{subfigure}{.5\textwidth}
  \centering
  \includegraphics[width=0.8\linewidth]{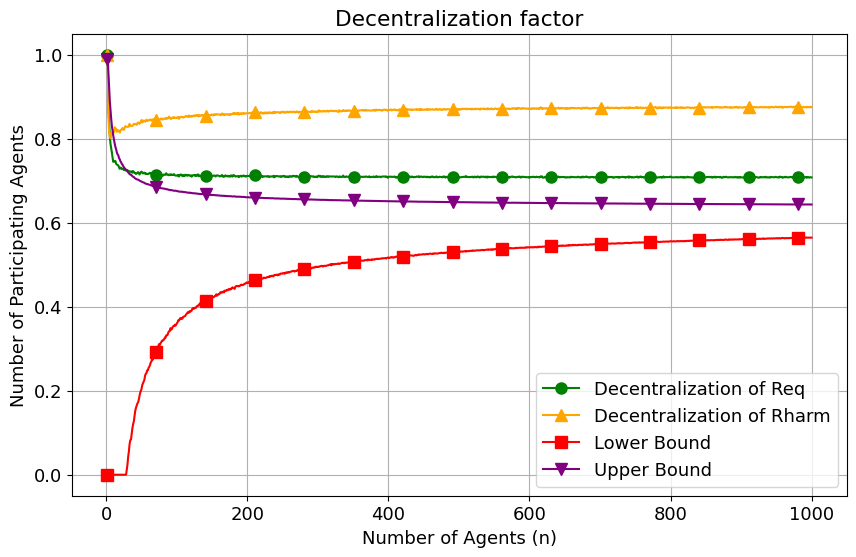}
  \caption{Decentralization factor}
  \label{fig:DecentrExp}
\end{subfigure}
\caption{Samples are drawn from Exponential  distribution with parameter $\lambda=1$ normalized to $(0,1)$ }
\label{fig:Exp}
\end{figure}

In the second set of experiments, we again sample agents' costs from Uniform and Exponential distributions and assign times according to the Uniform $[0,10]$ distribution.  To ensure that lower-cost agents take longer to complete the task, we now sample the costs and times separately and sort both lists. We first determine the number of participants using the $R^{\text{harm}}$ and then apply $\Mech^{\text{I-GSP}}$ to select the same number of participants. We record the fastest task completion time under each mechanism. The results confirm that  $\Mech^{\text{I-GSP}}$ consistently outperforms $R^{\text{harm}}$ in terms of the fastest completion time, as expected, since it prioritizes selecting the fastest agents while maintaining feasibility.

\begin{figure}[ht]
\centering
\begin{subfigure}{.5\textwidth}
  \centering
  \includegraphics[width=0.9\linewidth]{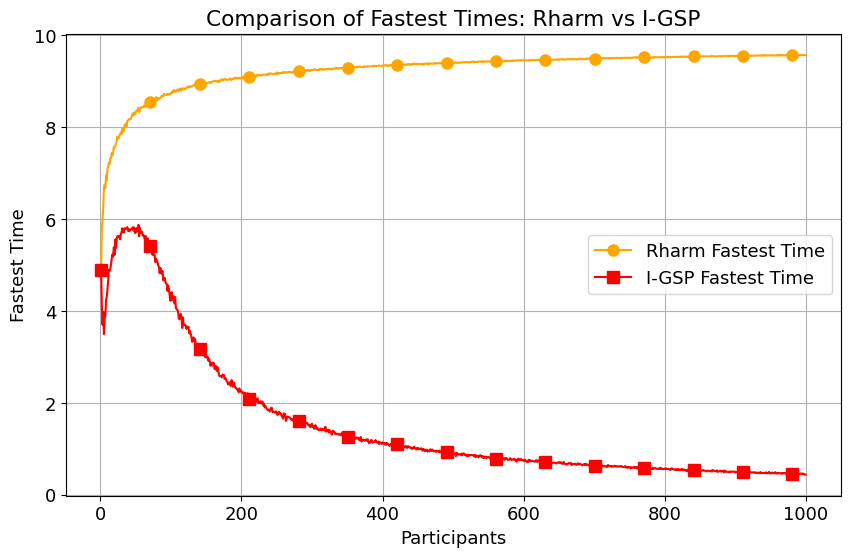}
  \caption{Costs are drawn form Uniform $(0,1)$}
  \label{fig:timeU}
\end{subfigure}%
\begin{subfigure}{.5\textwidth}
  \centering
  \includegraphics[width=0.9\linewidth]{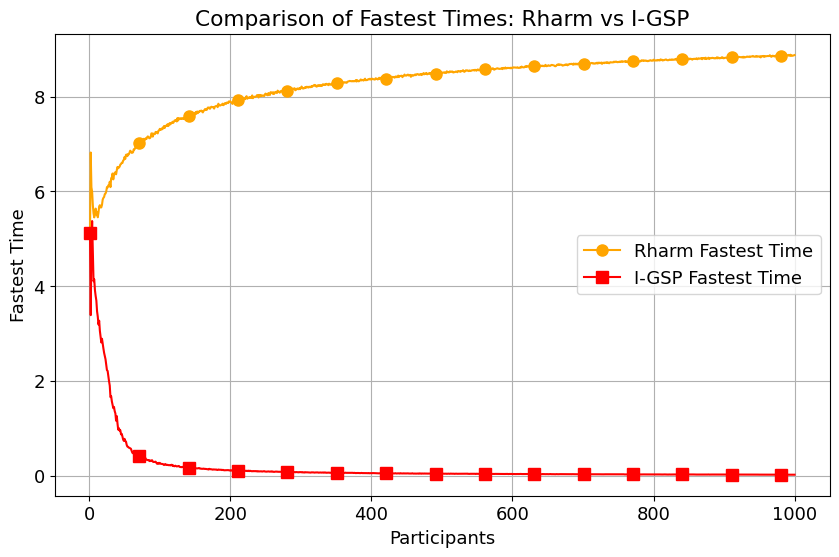}
  \caption{Costs are drawn from Exponential with $\lambda=1$}
  \label{fig:timeExp}
\end{subfigure}
\caption{Times are drawn from Uniform $(0,10)$ }
\label{fig:times}
\end{figure}

\section{Conclusion}\label{sec: conclusion}

We introduce a simple model for analyzing the trade-offs between decentralization and efficiency in verifiable computation. We completely characterize the power and limitations of revelation and non-revelation mechanisms in our model. Our results show provable advantages for revelation mechanisms, i.e., a so-called \emph{revelation gap}.

Our work points to several interesting future avenues. As a first step towards ensuring decentralization and efficiency in verifiable computation, we study the setting where there is a single request/client. The setting where solution providers simultaneously interact with (and budget their computational resources across) multiple clients' mechanisms is left as a future endeavor. 
A limitation of our model is the assumption that it is not possible for an agent to copy answers from other agents (i.e. claim rewards without doing any work). In the real-world, therefore, for our theoretical guarantees to go through, there needs to be a secure and authenticated channel (e.g. answers are signed and encrypted) that the agents can use to interact with the client. 
Our model assumes that an agent has a fixed cost that they pay to produce a solution at a fixed time; we leave it open to consider a setting where agents can incur a higher cost to provide a faster solution, but we conjecture that incentives have a milder effect in such a setting.  Finally,  we do not consider potential collusion among the solution providers.

\newpage
\bibliographystyle{alpha}
\bibliography{refs}

\newcommand{\etalchar}[1]{$^{#1}$}
\begin{thebibliography}{HKFTW23}

\bibitem[BKDV96]{baye1996all}
Michael~R Baye, Dan Kovenock, and Casper~G De~Vries.
\newblock The all-pay auction with complete information.
\newblock {\em Economic Theory}, 8:291--305, 1996.

\bibitem[Blo24]{bitcoinhashrate}
Blockchain.com.
\newblock Bitcoin hashrate distribution, 2024.
\newblock \url{https://www.blockchain.com/explorer/charts/pools}.

\bibitem[Cor07]{corchon2007theory}
Luis~C Corch{\'o}n.
\newblock The theory of contests: a survey.
\newblock {\em Review of economic design}, 11:69--100, 2007.

\bibitem[DA24]{statsweb3}
Dune-Analytics.
\newblock Monthly statistics of zk-pairing, 2024.
\newblock \url{https://dune.com/frits/zk-verify-pairing}.

\bibitem[ea23]{oracle_risks}
Mohamed~Damak et~al.
\newblock Utility at a cost: Assessing the risks of blockchain oracles.
\newblock S\&P Global, 2023.
\newblock \url{https://www.spglobal.com/_division_assets/images/special-editorial/blockchain-oracles/spgmi_utility-at-a-cost-v5.pdf}.

\bibitem[EOS07]{edelman2007internet}
Benjamin Edelman, Michael Ostrovsky, and Michael Schwarz.
\newblock Internet advertising and the generalized second-price auction: Selling billions of dollars worth of keywords.
\newblock {\em American economic review}, 97(1):242--259, 2007.

\bibitem[FH18]{feng2018end}
Yiding Feng and Jason~D Hartline.
\newblock An end-to-end argument in mechanism design (prior-independent auctions for budgeted agents).
\newblock In {\em 2018 IEEE 59th Annual Symposium on Foundations of Computer Science (FOCS)}, pages 404--415. IEEE, 2018.

\bibitem[FHL21]{feng2021revelation}
Yiding Feng, Jason~D Hartline, and Yingkai Li.
\newblock Revelation gap for pricing from samples.
\newblock In {\em Proceedings of the 53rd Annual ACM SIGACT Symposium on Theory of Computing}, pages 1438--1451, 2021.

\bibitem[Fla24a]{mevbuildermarket}
Flashbots.
\newblock Mev-boost analytics, 2024.
\newblock \url{https://www.relayscan.io/overview?t=7d}.

\bibitem[Fla24b]{mevboost}
Flashbots.
\newblock Overview of mev-boost, 2024.
\newblock \url{https://docs.flashbots.net/flashbots-mev-boost/introduction}.

\bibitem[FPV15]{friedman2015dynamic}
Eric Friedman, Christos-Alexandros Psomas, and Shai Vardi.
\newblock Dynamic fair division with minimal disruptions.
\newblock In {\em Proceedings of the sixteenth ACM conference on Economics and Computation}, pages 697--713, 2015.

\bibitem[GBE{\etalchar{+}}18]{gencer18}
Adem~Efe Gencer, Soumya Basu, Ittay Eyal, Robbert van Renesse, and Emin~G{\"u}n Sirer.
\newblock Decentralization in bitcoin and ethereum networks.
\newblock In Sarah Meiklejohn and Kazue Sako, editors, {\em Financial Cryptography and Data Security}, pages 439--457, Berlin, Heidelberg, 2018. Springer Berlin Heidelberg.
\newblock \url{https://doi.org/10.1007/978-3-662-58387-6_24}.

\bibitem[GGP10]{gennaro2010non}
Rosario Gennaro, Craig Gentry, and Bryan Parno.
\newblock Non-interactive verifiable computing: Outsourcing computation to untrusted workers.
\newblock In {\em Advances in Cryptology--CRYPTO 2010: 30th Annual Cryptology Conference, Santa Barbara, CA, USA, August 15-19, 2010. Proceedings 30}, pages 465--482. Springer, 2010.
\newblock \url{https://link.springer.com/chapter/10.1007/978-3-642-14623-7_25}.

\bibitem[GKCC13]{Gervais13}
Arthur Gervais, Ghassan Karame, Srdjan Capkun, and Vedran Capkun.
\newblock Is bitcoin a decentralized currency?
\newblock Cryptology ePrint Archive, Paper 2013/829, 2013.
\newblock \url{https://eprint.iacr.org/2013/829}.

\bibitem[HKFTW23]{heimbach2023ethereum}
Lioba Heimbach, Lucianna Kiffer, Christof Ferreira~Torres, and Roger Wattenhofer.
\newblock Ethereum's proposer-builder separation: Promises and realities.
\newblock In {\em Proceedings of the 2023 ACM on Internet Measurement Conference}, pages 406--420, 2023.
\newblock \url{https://doi.org/10.1145/3618257.3624824}.

\bibitem[MS01]{Moldovanu2001}
Benny Moldovanu and Aner Sela.
\newblock The optimal allocation of prizes in contests.
\newblock {\em The American Economic Review}, 91(3):542--558, 2001.

\bibitem[Net24]{gevulot}
Gevulot Network.
\newblock Gevulot, 2024.
\newblock \url{https://docs.gevulot.com/gevulot-docs/network/overview}.

\bibitem[nF24]{nil}
=nil; Foundation.
\newblock Proof market, 2024.
\newblock \url{https://docs.nil.foundation/proof-market/intro}.

\bibitem[OKK24]{karakostas2022sok}
Christina Ovezik, Dimitris Karakostas, and Aggelos Kiayias.
\newblock Sok: A stratified approach to blockchain decentralization.
\newblock In {\em Financial Cryptography and Data Security}, Lecture Notes in Computer Science. Springer, 2024.
\newblock \url{https://fc24.ifca.ai/}.

\bibitem[PHGR16]{parno2016pinocchio}
Bryan Parno, Jon Howell, Craig Gentry, and Mariana Raykova.
\newblock Pinocchio: Nearly practical verifiable computation.
\newblock {\em Communications of the ACM}, 59(2):103--112, 2016.
\newblock \url{https://doi.org/10.1145/2856449}.

\bibitem[PLD23]{oracle_survey}
Amirmohammad Pasdar, Young~Choon Lee, and Zhongli Dong.
\newblock Connect api with blockchain: A survey on blockchain oracle implementation.
\newblock {\em ACM Comput. Surv.}, 55(10), feb 2023.

\bibitem[Sta25]{statsweb2}
Statista.
\newblock Public cloud - worldwide, 2025.
\newblock \url{https://www.statista.com/outlook/tmo/public-cloud/worldwide}.

\bibitem[suc24]{succinct}
Succinct network: Prove the world’s software, 2024.

\bibitem[{Sup}23]{SupraVRF}
{Supra Research}.
\newblock Supra vrf service: Whitepaper.
\newblock Supra dVRF Product Documentation, 2023.
\newblock \url{https://supra.com/docs/SupraOracles-VRF-Service-Whitepaper.pdf}.

\bibitem[Swe22]{delphidigital}
Craig Sweney.
\newblock Chainlink verifiable random function activity surges.
\newblock Delphi Digital Report, 2022.
\newblock \url{https://members.delphidigital.io/reports/chainlink-verifiable-random-function-activity-surges}.

\bibitem[tai24]{taiko}
taiko.xyz.
\newblock Taiko tokenomics, 2024.
\newblock \url{https://github.com/taikoxyz/taiko-mono/blob/alpha-4/packages/protocol/docs/tokenomics_staking.md}.

\bibitem[WEY{\etalchar{+}}24]{BlockchainCensorship}
Anton Wahrst\"{a}tter, Jens Ernstberger, Aviv Yaish, Liyi Zhou, Kaihua Qin, Taro Tsuchiya, Sebastian Steinhorst, Davor Svetinovic, Nicolas Christin, Mikolaj Barczentewicz, and Arthur Gervais.
\newblock Blockchain censorship.
\newblock In {\em Proceedings of the ACM on Web Conference 2024}, page 1632–1643, 2024.

\bibitem[WZY{\etalchar{+}}25]{wang2024mechanism}
Wenhao Wang, Lulu Zhou, Aviv Yaish, Fan Zhang, Ben Fisch, and Benjamin Livshits.
\newblock Mechanism design for zk-rollup prover markets.
\newblock In {\em 29th International Conference on Financial Cryptography and Data Security}. Springer, 2025.

\end{thebibliography}

\end{document}